\documentclass[11pt,3p]{elsarticle}

\usepackage{graphicx}
\usepackage{epstopdf}
\usepackage{hyperref}
\usepackage{latexsym,amsmath,amssymb}
\usepackage{graphicx,epsfig}
\usepackage[normalem]{ulem}
\usepackage{algo}
\usepackage{amsthm}

\theoremstyle{plain}
\newtheorem{theorem}{Theorem}

\newtheorem{lemma}[theorem]{Lemma}

\theoremstyle{definition}
\newtheorem{definition}[theorem]{Definition}
\newtheorem{example}{Example}

\theoremstyle{remark}

\newcommand{\PNF}{\mathrm{PNF}}

\newcommand{\fmax}{F}

\newcommand{\select}{{\textit{select}}}
\newcommand{\rank}{{\textit{rank}}}
\newcommand{\RLE}{{\textit{rle}}}

\renewcommand{\epsilon}{\varepsilon}

\def\dom{\vartriangleright}
\def\domb{\blacktriangleright}

\def\bmin{{\textit bmin}}
\def\bmax{{\textit bmax}}

\begin{document}

\begin{frontmatter}

\title{Binary Jumbled String Matching for Highly Run-Length\\ Compressible Texts\tnoteref{note1}}
\tnotetext[note1]{A preliminary version of this paper was published on arXiv under the title ``Binary jumbled string matching: Faster indexing in less space'', \href{http://arxiv.org/abs/1206.2523v2}{arXiv:1206.2523v2}, 2012.}

\author{Golnaz Badkobeh}
\ead{golnaz.badkobeh@kcl.ac.uk}
\address{Department of Informatics, King's College, London, UK}

\author{Gabriele Fici}
\ead{gabriele.fici@unipa.it}
\address{Dipartimento di Matematica e Informatica, Universit\`a di Palermo, Italy}

\author{Steve Kroon}
\ead{kroon@sun.ac.za}
\address{Computer Science Division, Stellenbosch University, South Africa}

\author{Zsuzsanna Lipt\'ak}
\ead{zsuzsanna.liptak@univr.it}
\address{Dipartimento di Informatica, Universit\`a di Verona, Italy}

\begin{abstract}
The Binary Jumbled String Matching problem is defined as: Given a string $s$ over $\{a,b\}$ of length $n$ and a query $(x,y)$, with $x,y$ non-negative integers, decide whether $s$ has a substring $t$ with exactly $x$ $a$'s and $y$ $b$'s. 
Previous solutions created an index of size $O(n)$ in a pre-processing step, which was then used to answer queries in constant time. The fastest algorithms for construction of this index have running time $O(n^2/\log n)$ [Burcsi et al., FUN 2010; Moosa and Rahman, IPL 2010], or $O(n^2/\log^2 n)$ in the word-RAM model [Moosa and Rahman, JDA 2012].
We propose an index constructed directly from the run-length encoding of $s$. The construction time of our index is $O(n+\rho^2\log \rho)$, where $O(n)$ is the time for computing the run-length encoding of $s$ and $\rho$ is the length of this encoding---this is no worse than previous solutions if $\rho = O(n/\log n)$ and better if  $\rho = o(n/\log n)$. 
Our index $L$ can be queried in $O(\log \rho)$ time. While $|L|= O(\min(n, \rho^{2}))$ in the worst case, preliminary investigations have indicated that $|L|$ may often be close to $\rho$.
Furthermore, the algorithm for constructing the index is conceptually simple and easy to implement.
In an attempt to shed light on the structure and size of our index, we characterize it in terms of the prefix normal forms of $s$ introduced in  [Fici and Lipt\'ak, DLT 2011]. 
\end{abstract}


\end{frontmatter}

\newpage

\section{Introduction}\label{sec:intro}

Binary jumbled string matching is defined as follows: Given a string $s$ over $\{a,b\}$ and a query vector $(x,y)$ of non-negative integers $x$ and $y$, decide whether $s$ has a substring containing exactly $x$ $a$'s and $y$ $b$'s. If this is the case, we say that $(x,y)$ {\em occurs} in $s$. The {\em Parikh set of $s$}, $\Pi(s)$, is the set of all vectors occurring in $s$.

For one query, the problem can be solved optimally by a simple sliding window algorithm in $O(n)$ time, where $n$ is the length of the text. Here we are interested in the {\em indexing variant} where the text is fixed, and we expect a large number of queries. 
Recently, this problem and its variants have generated much interest ~\cite{CFL09,BCFL10,BCFL12_IJFCS,MR10,MR12,BCFL12_TOCS,CLWY12}.
The crucial observation is based on the following property of binary strings: 

\medskip
\noindent {\bf Interval Lemma.}
\emph{(\cite{CFL09}) Given a string $s$ over $\Sigma = \{a,b\}$, $|s|=n$. For every $m\in \{1,\ldots,n\}$: if,  for some $x<x'$, both $(x,m-x)$ and $(x',m-x')$ occur in $s$, then so does $(z,m-z)$ for all $z$, $x<z<x'$.
}
\medskip

\noindent It thus suffices to store, for every query length $m$, the minimum and maximum number of $a$'s in all $m$-length substrings of $s$. This information can be stored in a linear size index, and now any query of the form $(x,y)$ can be answered by looking up whether $x$ lies between the minimum and maximum number of $a$'s for length $m=x+y$. The query time is proportional to the time it takes to find $x+y$ in the index, which is constant in most implementations.

This index can be constructed naively in $O(n^2)$ time. 
In~\cite{BCFL10} and independently in~\cite{MR10}, construction algorithms were presented with running time $O(n^2/\log n)$, using reduction to min-plus convolution. In the word-RAM model, the running time can again be reduced to $O(n^2/\log^2 n)$, using bit-parallelism~\cite{MR12}. More recently, a Monte Carlo algorithm with running time $O(n^{1+\epsilon})$ was introduced~\cite{CLWY12}, which constructs an approximate index allowing one-sided errors, with the probability of an incorrect answer depending on the choice of $\epsilon$.

Any binary string $s$ can be uniquely written in the form $s = $ $a^{u_1}b^{v_1}a^{u_2}b^{v_2}\cdots a^{u_r}b^{v_r}$, where the $u_i,v_i$ are non-negative integers, all non-zero except  possibly $u_1$ and $v_r$. The {\em run-length encoding} of $s$ is then defined as $\RLE(s) = (u_1,v_1,u_2,v_2,\ldots,u_r,v_r)$. This representation is often used to compress strings, especially in domains where long runs of characters occur frequently, such as the representation of digital images,  multimedia data\-bases, and time series.

In this paper, we present the Corner Index $L$ which, for strings with good run-length compression, is much smaller than the linear size index used by all previous solutions. It is constructed directly from the run-length encoding of $s$, in time $O(\rho^2 \log \rho)$, where $\rho = |\RLE(s)|$.  The Corner Index has worst-case size $\min(n,\rho^2)$ (measured in the number of entries, which fit into two computer words). We pay for this with an increase in lookup time from $O(1)$ to $O(\log |L|)=O(\log \rho)$. 

In a recent paper~\cite{FL11}, the {\em prefix normal forms} of a binary string were introduced. Given $s$ of length $n$, $\PNF_a(s)$ is the unique string such that, for every $1\leq m \leq n$,  its $m$-length prefix has the same number of $a$'s as the maximum number of $a$'s in any $m$-length substring of $s$; $\PNF_b(s)$ is defined analogously. It was shown in~\cite{FL11} that two strings $s$ and $t$ have the same Parikh set if and only if $\PNF_a(s) = \PNF_a(t)$ and $\PNF_b(s) = \PNF_b(t)$. From this perspective, our index can be viewed as storing the run-length encodings of $\PNF_a(s)$ and $\PNF_b(s)$.  This allows us a fresh view on the problem, and may point to a promising way of proving bounds on the index size. Moreover, our algorithm constitutes an improvement both for the computation and the testing problems on prefix normal forms (see~\cite{FL11}) whenever $\RLE(s)$ is short.

The construction time of $O(n+\rho^2 \log \rho)$, where $O(n)$ is for computing $\RLE(s)$ and $O(\rho^2 \log \rho)$ for constructing the Corner Index, is much better than the previous $O(n^2/\log n)$ time algorithms for strings with short run-length encodings, and no worse as long as $\rho=O(n/\log n)$. For strings with good run-length compression, the increase in lookup time from $O(1)$ to $O(\log |L|)$ is justified in our view by the reduced size and construction time of the new index. 
Finally, our algorithm is conceptually simple and easy to implement.


\section{Preliminaries}\label{sec:preliminaries}

A binary string $s = s_1s_2\cdots s_n$ is a finite sequence of characters from $\{a,b\}$. We denote the length of $s$ by $|s|$. For two strings $s,t$, we say that $t$ is a substring of $s$ if there are indices $1\leq i, j\leq |s|$ such that \ $t = s_i\cdots s_j$. If $i=1$, then $t$ is called a prefix of $s$. We denote by $|s|_a$ (resp.\ $|s|_b$) the number of $a$'s (resp.\ $b$'s) in $s$. The {\em Parikh vector} of $s$ is defined as  $p(s) = (|s|_a,|s|_b)$. We say that a Parikh vector $q$ {\em occurs} in string $s$ if $s$ has a substring $t$ such that $p(t) = q$. The {\em Parikh set} of $s$, $\Pi(s)$, is the set of all Parikh vectors occurring in $s$.

The Interval Lemma from the Introduction implies that, for any binary string $s$, there are functions $F$ and $f$ s.t.\ 
\begin{equation}\label{eq:F}
(x,y) \text{ occurs in $s$ if and only if } f(x+y) \leq x\leq F(x+y),
\end{equation}

\noindent namely, for $m=0,\ldots,|s|$,  $F(m) = \max \{ x \mid (x,m-x)\in \Pi(s) \}$ and $f(m) = \min \{ x \mid (x,m-x)\in \Pi(s) \}$.
This can be stated equivalently in terms of the minimum and maximum number of $b$'s in all substrings containing a {\em fixed number} $i$ of $a$'s. Let us denote by $\bmin(i)$ (resp.\ $\bmax(i)$) the minimum (resp.\ maximum) number of $b$'s in a substring containing exactly $i$ $a$'s. Then 
\begin{equation}\label{eq:G}
(x,y) \text{ occurs in $s$ if and only if } \bmin(x) \leq y\leq \bmax(x).
\end{equation}

The table of functions $F$ and $f$ in~\eqref{eq:F} is the index used in most algorithms for Binary Jumbled String Matching~\cite{BCFL12_IJFCS,MR10,MR12}, while that of functions $\bmin$ and $\bmax$ in~\eqref{eq:G} was used in~\cite{BCFL12_TOCS}. Even though the latter is always smaller, both are linear size in $n$. Note that one table can be computed from the other in linear time (e.g.\ $\bmin(i) = \min \{ m \mid F(m) = i \} - i$).

\begin{example}
Let $s = aabababbaaabbaabbb$. Then $(3,3)$ occurs in $s$ 
while $(5,1)$ does not. We have $F(6)=4$ and $f(6)=2$, $\bmin(3)=0$ and $\bmax(3)=5$. We give the full table of values of the two functions $\bmin$ and $\bmax$ in Table~\ref{tab:G}.

\begin{table}\centering  
\begin{small}
\begin{raggedright}
\begin{tabular}{c *{30}{@{\hspace{2.1mm}}l}}
 $i$  &  0\hspace{1ex}  & 1\hspace{1ex} & 2\hspace{1ex} & 3\hspace{1ex} & 4\hspace{1ex} & 5\hspace{1ex} & 6\hspace{1ex} & 7\hspace{1ex} &
8\hspace{1ex} & 9\hspace{1ex} \\
\hline \rule[-6pt]{0pt}{22pt}
$ \bmin(i)$ & 0 & 0    & 0& 0& 2& 2& 4& 4& 6& 6  \\
\hline \rule[-6pt]{0pt}{22pt}
$ \bmax(i)$ &  3    & 3& 5& 5& 5& 7& 8& 9& 9& 9\\
\hline \\
\end{tabular}
\end{raggedright}\caption{\label{tab:G}Functions $\bmin$ and $\bmax$ for the string $s = aabababbaaabbaabbb$.}
\end{small}
\end{table}

\end{example}


\section{The Corner Index}\label{sec:index}

In Fig.~\ref{fig:Ggraph}, we plot both functions $\bmin$ and $\bmax$ for our example string. The $x$-axis denotes the number of $a$'s and the $y$-axis the number of $b$'s. It follows from \eqref{eq:G} that the integer points within the shaded area correspond to the Parikh set of $s$. The crucial observation is: Since both functions $\bmin$ and $\bmax$ are monotonically increasing step functions, it is sufficient to store those points where they increase. These points are specially marked in Fig.~\ref{fig:Ggraph}.

\begin{figure}
\begin{center}
\includegraphics[height=65mm]{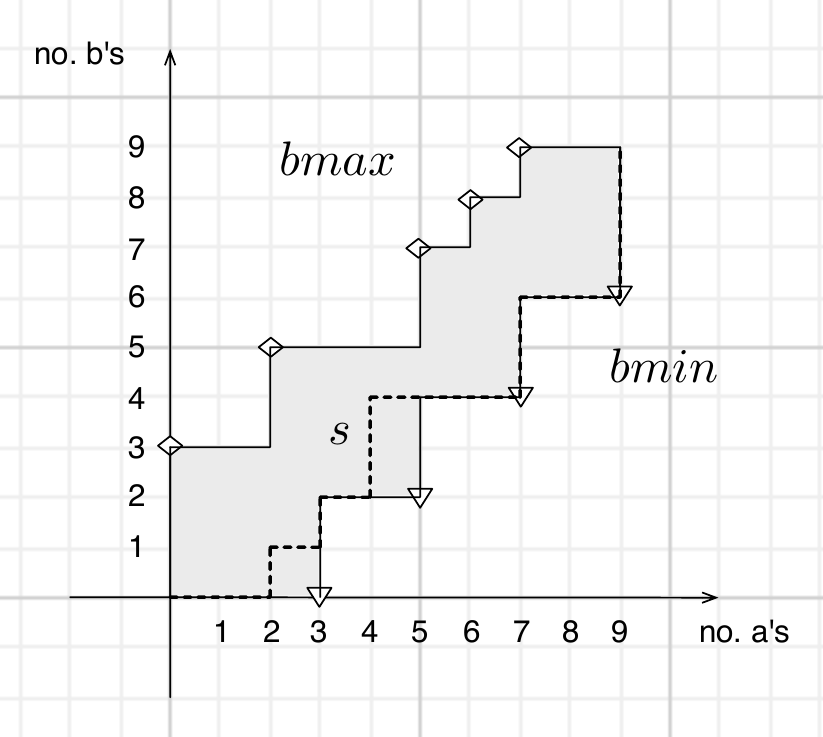}
\end{center}
\caption{\label{fig:Ggraph}The $\bmin$ and $\bmax$ functions for the string $s= aabababbaaabbaabbb$. 
Representing binary strings as walks on the integer grid, $s$ is indicated by a dashed line, while the functions $\bmin$ and $\bmax$ correspond to the prefix normal forms of $s$; see Sec.~\ref{sec:PNF} for more details.}
\end{figure}

\begin{example}\label{ex:4}  
In our example, these points are, for $\bmin$: $\{ (3,0), (5,2), (7,4), (9,6)\},$ and for $\bmax$: $\{ (0,3), (2,5), (5,7), (6,8), (7,9)\}$.

%
\end{example}

\begin{definition}
We define the {\em Corner Index} for the Parikh set of a given binary string $s$ as two ordered sets $L_{\min}$ and $L_{\max}$, where
\begin{align}
L_{\min} &= \{ (i,\bmin(i)) \mid i=|s|_a \text{ or } \bmin(i) < \bmin(i+1) \},\\
L_{\max} &= \{ (i,\bmax(i)) \mid i=0 \text{ or } \bmax(i) > \bmax(i-1) \}.
\end{align}
\end{definition}

The order is according to both components, since for any $(x,y),$ $(x',y')\in L_{\min}$ (or $\in L_{\max}$), we have that $x<x'$ if and only if $y<y'$. Now for any $x$, we can recover $\bmin(x)$ from $L_{\min}$ (resp.\ $\bmax(x)$ from $L_{\max}$) by noting that 
\begin{align}\label{eq:xR}
\bmin(x) &= \bmin(x_R), & \quad \text{} x_R = \min \{ x' \mid x' \geq x, \exists y': (x',y') \in L_{\min} \},\\
\bmax(x) &= \bmax(x_L), & \quad \text{} x_L = \max \{ x' \mid x' \leq x, \exists y': (x',y') \in L_{\max}\}.
\end{align}

\subsection{Construction}

To construct the Corner Index, we will use the run-length encoding of $s$, $\RLE(s) = (u_1,v_1,u_2,v_2,\ldots,u_r,v_r)$. We refer to maximal substrings consisting only of $a$'s (resp.\ $b$'s) as $a$-runs (resp.\ $b$-runs), and denote by $\rho=|\RLE(s)|$, thus $2r-2 \leq \rho \leq 2r$.  It follows directly from the definitions that
\begin{equation}\label{eq:G2}
(x,y) \in \Pi(s) \quad \Rightarrow \quad \forall x'\leq x: \bmin(x') \leq y \text{ and } \forall x'\geq x: \bmax(x') \geq y.
\end{equation}

\begin{lemma}\label{lemma:corners}
Let $s$ be a binary string and $(x,y)\in \Pi(s)$. Then there exists a substring $t$ of $s$ which begins and ends with a full $a$-run such that $p(t) = (x_1,y_1)$ and $x_1\geq x, y_1\leq y$.  Similarly, there is a substring $t'$ of $s$ which begins and ends with a full $b$-run such that $p(t') = (x_2,y_2)$ and $x_2\leq x, y_2\geq y$.
\end{lemma}

\begin{proof}
 Let $w=s_{i}\cdots s_{j}$ be a substring of $s$ such that
$p(w)=(x,y)$. 
If $s_i=a$,
then extend $w$ to the left to the beginning of the 
$a$-run containing $s_i$; if $s_i=b$, then shrink $w$ from the left to exclude
all $b$'s of the $b$-run containing $s_i$, likewise for $s_j$. 
The
substring $t$
so obtained fulfils the requirements. A substring $t'$ can be found analogously by extending
$b$-runs and shrinking $a$-runs.
 \end{proof}

Lemma~\ref{lemma:corners}, together with \eqref{eq:G2}, implies that it suffices to compute  substrings beginning and ending with full $a$-runs (for $L_{\min}$) and beginning and ending with full $b$-runs (for $L_{\max}$). The algorithm generates the Parikh vectors of these substrings one by one, inspects them and incrementally constructs $L_{\min}$ and $L_{\max}$. For brevity of exposition, we only give the algorithm for constructing $L_{\min}$; $L_{\max}$ can be computed simultaneously. We need the following definition:

\begin{definition}
Let $(x,y), (x',y') \in \Pi(s)$. We say that $(x,y)$ {\em dominates} $(x',y')$, denoted $(x,y) \dom (x',y')$, if $(x,y) \neq (x',y')$, $x\geq x'$ and $y\leq y'$. For $X\subset \Pi(s),(x,y)\in \Pi(s)$, define $X\dom (x,y)$ iff exists $(x',y') \in X$ s.t. $(x',y') \dom (x,y)$.
\end{definition}

Since $\dom$ is irreflexive and transitive, it is a (strict) partial order. Note that $L_{\min}$ is the set of maximal elements in the poset $(\Pi(s), \dom)$. (Another relation $\domb$ can be defined analogously s.t.\ the set of maximal elements equals $L_{\max}$.)

We present the algorithm computing the index in Fig.~\ref{fig:algo}. Recall that $u_i$ (resp.~$v_i$) is the length of the $i$-th $a$-run (resp.~$b$-run) of $s$. We compute, for every interval size $k \geq 1$, and every $1\leq i\leq r-k+1$, the Parikh vector $(x,y)$ of the substring starting with the $i$th $a$-run and spanning $k$  $a$-runs, and query whether $(x,y)$ is dominated by some element in $L_{\min}$. Note that this is the case if and only if $(x,y)$ is dominated by the unique pair $(x',y')$ where $x'=x_R$ from~\eqref{eq:xR}. If no element of $L_{\min}$ dominates $(x,y)$, then it is added to $L_{\min}$, and all elements of $L_{\min}$ which $(x,y)$ dominates are removed from the list. These can be found consecutively in decreasing order from the position where $(x,y)$ was inserted. The algorithm is illustrated in Fig.~\ref{fig:LG} on our example string. Top left gives the run-length encoding of $s$, with $a$-runs in the first row, and $b$-runs in the second. On the right we list all pairs which need to be inspected, and on the left bottom the final list $L_{\min}$. Elements which are inserted into $L_{\min}$ and later removed are struck through.

\begin{figure}
\begin{center}
\begin{algorithm}{Construct $L_{\min}$}{\label{algo:construct}}
{\bf input:} $\RLE(s) = (u_1,v_1,u_2,v_2,\ldots,u_r,v_r)$\\
\qfor $k$ from $1$ to $r$\\
\qfor $i=1$ to $r-k+1$\\
$(x,y) \gets (u_{i}+\ldots +u_{i+k-1}, v_{i}+\ldots +v_{i+k-2})$\\
\qif not $L_{\min} \dom (x,y)$ \\
\qthen 
insert $(x,y)$ into $L_{\min}$ \\
\qfor each $(x',y')$ in $L_{\min}$ s.t.~$(x,y) \dom (x',y')$, \\
 delete $(x',y')$ from $L_{\min}$
\qfi
\qend
\qend
\end{algorithm}
\end{center}
\caption{\label{fig:algo}The algorithm computing $L_{\min}$.}
\end{figure}

\begin{figure}
\begin{minipage}{6.5cm}
$$
\begin{array}{c@{\;\;}|@{\;\;}rrrrr}
a & 2 & 1 & 1 & 3 & 2\\
b & 1 & 1 & 2 & 2 & 3
\end{array}
$$
\vspace{.1cm}
$$L_{\min}: \text{\sout{\ensuremath{(2,0)}}},(3,0),\text{\sout{\ensuremath{(4,2)}}}, (5,2),$$ \\[-9mm]  $$\hspace{8mm}\text{\sout{\ensuremath{(6,4)}}}, (7,4), (9,6)$$
\end{minipage}
\begin{minipage}{7cm}
$
\begin{array}{l @{\hspace{.5cm}} l}
k & \\
\hline
1 & (2,0) (1,0) (1,0) (3,0) (2,0) \\
2 & (3,1) (2,1) (4,2) (5,2) \\
3 & (4,2) (5,3) (6,4) \\
4 & (7,4) (7,5) \\
5 & (9,6)
\end{array}
$
\end{minipage}
\caption{Computation of $L_{\min}$ for the example $s = aabababbaaabbaabbb$.\label{fig:LG}}
\end{figure}

\subsection{Analysis}

The number of entries of each list is upper bounded by $\min \{|s|_a, |s|_b,{r+1 \choose 2}\}$, thus the total size of the Corner Index is $O(\min(n,\rho^2))$. 
The query time is $O(\log |L|) = O(\log \rho)$. 

The working space of the construction algorithm is the maximum size $L_{\min}$ reaches during the algorithm, which is at most  ${r+1\choose 2} = O(\rho^2)$. 
For the construction time, note that $O(\rho^2)$ pairs have to be inspected. For each, we have to decide whether it is dominated by an element in $L_{\min}$; this query amounts to finding $x_R$ from~\eqref{eq:xR} in $L_{\min}$, in
$O(\log \rho)$ time. Insertion of an element can cause more than one deletion in the list; however, since each element is deleted at most once, we have amortized time $O(\log \rho)$ per element, and thus altogether $O(\rho^2\log \rho)$ time for the construction algorithm. 

Note that $L_{\min}$ can be constructed by inspecting the ${r+1\choose 2}$ pairs in an arbitrary order, although our bound on the construction time assumes that the pairs are generated in constant time. We summarize:

\begin{theorem}
Queries for the Binary Jumbled String Matching problem can be answered in $O(\log \rho)$ time, using an index of size $O(\min(\rho^2,n))$, where $n$ is the length of the text and $\rho$ the length of its run-length encoding. The index can be constructed in  $O(n+ \rho^2 \log \rho)$ time from the string $s$.
\end{theorem}


\section{Prefix Normal Forms}\label{sec:PNF}

We recall the definitions of $\rank$ and $\select$ (cf.~\cite{NavMaek07}).
Given a binary string $s$, we denote, for $c\in \{a,b\}$, by $\rank_c(s,i) = |s_1\cdots s_i|_c$,  the number of $c$'s in the prefix of length $i$ of $s$, and by $\select_c(s,i)$ the position of the $i$'th $c$ in $s$, i.e., $\select_c(s,i) = \min\{ k : |s_1\cdots s_k|_c = i\}$. 

It is possible \cite{FL11} to associate to any binary string $s$ a unique string $s'$ such that for all $0\leq i\leq |s|$, $\fmax_{s}(i) = \fmax_{s'}(i) = \rank_a(s',i)$, i.e., for any length $i$, the number of $a$'s in the prefix of $s'$ of length $i$ equals the maximum number of $a$'s in any substring of $s$ of length $i$. The string $s'$ is called the \emph{prefix normal form} of $s$ with respect to $a$, denoted $\PNF_a(s)$; the prefix normal form w.r.t.~$b$, $\PNF_b(s)$, is defined analogously.

\begin{example}
For our string  $s = aabababbaaabbaabbb$, the prefix normal forms are  
\begin{align}
\PNF_{a}(s) &=aaabbaabbaabbaabbb, \quad \text{and}\\
\PNF_{b}(s) &=bbbaabbaaabbababaa.
\end{align}

 \end{example}
 
By definition, $\bmin(i)$ is the minimum number of $b$'s in a prefix of $\PNF_{a}(s)$ containing exactly $i$  $a$'s, 
and $\bmax(i)$ the maximum number of $b$'s in a prefix of $\PNF_{b}(s)$ containing exactly $i$  $a$'s. 
So we have: 
\begin{align}
F(i) &= \rank_a(\PNF_a(s),i) \quad \quad \quad \mbox{ for $0\leq i \leq |s|$,}\\
 \bmin(i) &=\select_{a}(\PNF_{a}(s),i)-i \quad \mbox{ for $0\le i \le |s|_{a}$}, \text{and} \\
\bmax(i) &= 
\begin{cases} \select_{a}(\PNF_{b}(s),i+1)-(i+1) & \mbox{ for $0\le i < |s|_{a}$}, \\ |s|_{b} & \mbox{ for $ i = |s|_{a}$}. 
\end{cases}
\end{align}
 
In fact, if we represent binary strings by drawing a horizontal unit line segment for each $a$ and a vertical one for each $b$, then 
$\PNF_a(s)$ is represented by function $\bmin$, and $\PNF_b(s)$ by function $\bmax$, see Fig.~\ref{fig:Ggraph}. 

Moreover, the run-length encoding of $\PNF_{a}(s)$ contains the same information as the list $L_{\min}$ output by our algorithm:  Indeed, let $\RLE(\PNF_{a}(s))$ $= (u'_1,v'_1,u'_2,v'_2,\ldots,u'_{r'},v'_{r'})$. Then, setting $p_{m}=\sum_{i=1}^{m}u'_{i}$ and $q_{m}=\sum_{i=1}^{m}v'_{i}$, one has 

\begin{equation}
L_{\min}=\{(p_{m},q_{m-1}) \mid m=1,\ldots ,r'\}.
\end{equation}

In particular, $|L_{\min}|=\frac{1}{2}|\RLE(\PNF_{a}(s))|$, and this gives a bound on the size of the output in terms of the prefix normal form.


\section{Open problems}\label{sec:conc}

We conclude with some open problems. First, we are interested in tighter bounds on the size of the Corner Index in terms of $\rho$, the length of the run-length encoding of the input string---our preliminary experiments on random strings suggest that the size of the index may often be close to $\rho$. Second, how much working space is required by our algorithm: in our experiments it was rare for the maximal size of the index during construction to exceed the final index size. Hopefully this working space can be bounded by making use of the structure of the posets $(\Pi(s), \dom)$ and $(\Pi(s),\domb)$ introduced in Sec.~\ref{sec:index}. Third, does the number of maximal pairs in these posets (which is the total length of the run-length encodings of the two $\PNF$s) constitute a lower bound on the size of {\em any} index for the Binary Jumbled String Matching problem? Better understanding these posets could also lead to an improvement of our algorithm's running time: if we could characterize maximal pairs, it may no longer be necessary to inspect all possible pairs.

\subsection*{Acknowledgements}

We thank Karim Dou\"ieb who took part in the original conception of this algorithm as well as in some of the early discussions, and two anonymous referees for helpful comments. Part of this research was done while Zs.L.~was supported by a Marie Curie Intra-European Fellowship (IEF) within the 7th Framework Programme of the European Union (grant no.~PIEF-GA-2010-274778).

\begin{small}
\bibliographystyle{abbrv}
\bibliography{parikh_2012}
\end{small}

\end{document}